\documentclass{article}
\usepackage{amsfonts,color,framed,amsthm}
\usepackage{amsmath}

\def\C{\ensuremath{\mathbb C}}
\def\N{\ensuremath{\mathbb N}}
\def\Z{\ensuremath{\mathbb Z}}
\newcommand{\be}{\begin{equation}}
\newcommand{\ee}{\end{equation}}
\newcommand{\bea}{\begin{eqnarray}}
\newcommand{\eea}{\end{eqnarray}}

\definecolor{shadecolor}{rgb}{0.97, 0.97, .97}

\newtheorem{theorem}{Theorem}
\newtheorem{definition}{Definition}
\newtheorem{proposition}{Proposition}
\newtheorem{lemma}{Lemma}
\newtheorem{corollary}{Corollary}
\theoremstyle{definition}
\newtheorem{remark}{Remark}

\title{Discrete matrix models for partial sums of conformal blocks associated to Painlev\'e transcendents}
\author{F. Balogh\\
\\
Scuola Internazionale Superiore di Studi Avanzati\\
via Bonomea 265, 34136 Trieste, Italy}

\begin{document}
\maketitle
\begin{abstract} A recently formulated conjecture of Gamayun, Iorgov and Lisovyy gives an asymptotic expansion of the Jimbo--Miwa--Ueno isomonodromic $\tau$-function for certain Painlev\'e transcendents. The coefficients in this expansion are given in terms of conformal blocks of a two-dimensional conformal field theory, which can be written as infinite sums over pairs of partitions. In this note a discrete matrix model is proposed on a lattice whose partition function can be used to obtain a multiple integral representation for the length restricted partial sums of the Painlev\'e conformal blocks. This leads to expressions of the partial sums involving H\"ankel determinants associated to the discrete measure of the matrix model, or equivalently, Wronskians of the corresponding moment generating function which is shown to be of the generalized hypergeometric type.
\end{abstract}
\noindent
{\bf Keywords:} Painlev\'e transcendents, isomonodromic $\tau$-functions, discrete matrix models, generalized hypergeometric functions
\section{Introduction and statement of results}
The six Painlev\'e equations $P_{\mathrm{I}} - P_{\mathrm{VI}}$ can be obtained from the isomonodromic deformation equations of $2\times 2$ matrix-valued linear differential equations with rational coefficients \cite{FIKN}.
Recently it was conjectured \cite{conformal_painleve} that  the Jimbo--Miwa--Ueno isomonodromic $\tau$-function \cite{JMU} associated to $P_{\mathrm{VI}}$ can be interpreted as a four-point correlator appearing in a two-dimensional conformal field theory, and this observation was extended in \cite{instanton_combinatorics} to include similar descriptions of $P_{\mathrm{V}}$ and $P_{\mathrm{III}}$ (see also \cite{painleve_connection_conformal_blocks,painleve_conformal_constr_appr, isomonodromic_liouville,its_lisovvy_tykhyy,bershtein_shchechkin}). One of the interesting aspects of these conjectures is that they would imply an asymptotic expansion of the $\tau$-function, in the limit when certain poles in the corresponding linear system coalesce. The asymptotic expansion of the $\tau$-function is written in terms of conformal blocks, which, by the AGT correspondence \cite{AGT}, can be expressed as infinite sums over pairs of partitions.

The aim of this paper is to give a closed formula for length-restricted partial sums (to be defined below) of the Painlev\'e conformal blocks. To define these quantities of interest, let $\vec a$ be a list of complex numbers
\be
\vec a = [a_1,\dots, a_p]
\ee
with $p \geq 0$, and consider the following $\vec{a}$-dependent meromorphic functions in the complex plane:
\begin{align}
\label{eq:P}
P(\vec a;z)&:=\prod_{i=1}^{p}(z+a_i)\\
\label{eq:Ga}
\Gamma(\vec a;z) &:= \prod_{i=1}^{p}\Gamma(z+a_i)\\
\label{eq:G}
G(\vec a;z)&:= \prod_{i=1}^{p}G(z+a_i)\ ,
\end{align}
where $\Gamma(z)$ stands for the gamma function and $G(z)$ is Barnes' $G$-function
\be
G(z+1)=(2\pi)^{z/2}\exp\left(-\frac{1}{2}z(z+1)-\frac{1}{2}\gamma z^{2}\right)\prod_{k=1}^{\infty}\left(\left(1+\frac{z}{k}\right)^{k}\exp\left(-z+\frac{z^2}{2k}\right)\right)
\ee
(also known as double gamma function), satisfying
\be
G(z+1)= \Gamma(z)G(z) \quad \mbox{and} \quad G(1)=1\ .
\ee
By convention, for $p=0$ (if the list is empty), the expressions \eqref{eq:P} to \eqref{eq:G} are defined to be identically $1$, as functions of the complex variable $z$.

Define also the quantities
\be
C(\vec a; \sigma) := \frac{G(\vec a; 1+\sigma)G(\vec a; 1-\sigma)}{G(1+2\sigma)G(1-2\sigma)}\ ,
\ee
that depend meromorphically on the parameter $\sigma$. 

Recall that the weight and the length of a partition $\lambda = (\lambda_1, \lambda_2, \dots)$ are
\be
|\lambda|=\sum_{i}\lambda_i\quad \mbox{and}\quad \quad \ell(\lambda) = \max\{i: \lambda_i >0\}\ ,
\ee
respectively. For a pair of partitions $\lambda$ and $\mu$, let
\begin{align}
{\mathcal B}_{\lambda, \mu}(\vec a;\sigma)
\nonumber
&:=\prod_{(i,j)\in \lambda}\frac{P(\vec a; i-j+\sigma)}{h_{\lambda}^2(i,j)(\lambda'_j+\mu_i-i-j+1+2\sigma)^2}\times\\
\label{eq:B_lambda_mu}
&\quad \times \prod_{(i,j)\in \mu}\frac{P(\vec a; i-j-\sigma)}{h_{\mu}^2(i,j)(\lambda_i+\mu'_j-i-j+1-2\sigma)^2}\ ,
\end{align}
where $\lambda'$ stands for the transposed partition, $(i,j)\in \lambda$ means that $j \leq \lambda_i$, and $h_{\lambda}(i,j)$ is the \emph{hook length}
\be
h_{\lambda}(i,j) = \lambda_i-i+\lambda'_j-j+1
\ee
associated to the box $(i,j)$ in the Young diagram of $\lambda$.

\begin{remark}
\label{rem:decomp}
Note that the coefficient ${\mathcal B}_{\lambda, \mu}(\vec a;\sigma)$ can be decomposed as
\be
{\mathcal B}_{\lambda, \mu}(\vec a;\sigma) = {\mathcal B}_{\lambda, \mu}(\vec [];\sigma)\prod_{(i,j)\in \lambda}P(\vec a; i-j+\sigma)\prod_{(i,j)\in \mu}P(\vec a; i-j-\sigma)\ ,
\ee
i.e., the coefficient ${\mathcal B}_{\lambda, \mu}(\vec a;\sigma)$ can be thought of as a ``dressed" version of the ``bare" coefficient ${\mathcal B}_{\lambda, \mu}(\vec [];\sigma)$.
\end{remark}

The isomonodromic $\tau$-functions for the Painlev\'e equations $P_{\mathrm{VI}}$, $P_{\mathrm{V}}$, $P_{\mathrm{III_{1}'}}$, $P_{\mathrm{III_{2}'}}$, $P_{\mathrm{III_{3}'}}$ are conjectured to have expansions of the general form
\be
\label{eq:tau_expansion}
\tau(t) = f(t)\sum_{n \in \Z}C(\vec a;\sigma+n)s^{n}t^{(\sigma+n)^2}{\mathcal B}(\vec a;\sigma+n; t)\ ,
\ee
 as $t \to 0$. The function ${\mathcal B}(\vec a;\sigma; t)$ is obtained using conformal field theory techniques and it is referred to as a conformal block. The pre-factor $f(t)$ and the list of parameters $\vec a$ depend on the Painlev\'e equation in question, as shown in Table \ref{table:1} (see \cite{instanton_combinatorics} for the details). The parameters $\theta_i$ are associated to the monodromy of the underlying $2\times 2$ Fuchsian ODE, ($\pm\theta_j$ are the eigenvalues of the residue matrices at the simple poles of the equation). The parameters $\sigma$ and $s$ can be thought of as integration constants associated to a given Painlev\'e solution.
\begin{table}[h!]
\begin{tabular}{|c|c|c|}
\hline
\parbox{2.5cm}{\begin{center}Painlev\'e\\transcendent\end{center}} & $\vec{a}$ &$f(t)$\\
\hline
$P_{\mathrm{VI}}$ & $[\theta_t-\theta_0,\theta_t+\theta_0,\theta_1-\theta_{\infty},\theta_1+\theta_{\infty}]$ & $(1-t)^{2\theta_t\theta_1}t^{-\theta_0^2-\theta_t^2}$\\
\hline
$P_{\mathrm{V}}$ & $[\theta_{*},\theta_0-\theta_t,-\theta_0-\theta_t]$ & $e^{-\theta_t t}$\\
\hline
$P_{\mathrm{III_{1}'}}$ & $[\theta_{*},\theta_{\star}]$ & $e^{-\frac{t}{2}}$\\
\hline
$P_{\mathrm{III_{2}'}}$ & $[\theta_{*}]$ & $1$\\
\hline
$P_{\mathrm{III_{3}'}}$ & $[\ ]$ & $1$\\
\hline
\end{tabular}
\label{table:1}
\caption{The pre-factor $f(t)$ and the parameters $\vec{a}$ for the different Painlev\'e equations considered in \cite{instanton_combinatorics}}
\end{table}

 By using the AGT correspondence, the Painlev\'e conformal blocks
can be expressed as infinite sums over pairs of partitions of the form
\be
{\mathcal B}(\vec a;\sigma; t):=\sum_{\lambda, \mu}{\mathcal B}_{\lambda, \mu}(\vec a;\sigma)t^{|\lambda|+|\mu|}\ .
\ee
It is of particular interest \cite{its_lisovvy_tykhyy} to find closed expressions for ${\mathcal B}(\vec a;\sigma; t)$
or to understand the asymptotic behavior of its partial sums.
The goal of the present paper is to show that the \emph{$K$-restricted partial sum}
\be
{\mathcal B}_{K}(\vec a;\sigma; t):=\sum_{\lambda, \mu \atop \ell(\lambda)\leq K, \ell(\mu) \leq K}{\mathcal B}_{\lambda, \mu}(\vec a;\sigma)t^{|\lambda|+|\mu|}\ ,
\ee
which itself is a sum of infinitely many terms, can be written in terms of an associated discrete matrix model. Our main result is based on the following crucial observation:
\begin{shaded}
\begin{theorem}
\label{thm:b}
Let $\lambda$ and $\mu$ be a pair of partitions such that
\be
\ell(\lambda)\leq K\ , \quad \ell(\mu) \leq K
\ee
for some positive integer $K$. Then
\be
\label{eq:b}
{\mathcal B}_{\lambda, \mu}(\vec a;\sigma)= Q(\vec a, K, \sigma)\prod_{1\leq i <j \leq 2K}(x_i-x_j)^{2}\prod_{i=1}^{2K}w(x_i)\ ,
\ee
where
\be
\begin{split}
x_i &= \lambda_i-i+K-\sigma\\
x_{K+i}&= \mu_i-i+K+\sigma 
\end{split}\qquad 1\leq i \leq K\ ,
\ee
the weight $w(z)$ is given by
\be
w(z)\equiv w_{\vec a,K, \sigma}(z):=\frac{1}{\Gamma(\vec a;K-z)\Gamma(z+1+\sigma)^2\Gamma(z+1-\sigma)^2}\ ,
\ee
and the factor $Q(\vec a, K, \sigma)$ has the explicit form
\be
Q(\vec a, K, \sigma) := \frac{G(\vec a;K+1+\sigma)G(\vec a;K+1-\sigma)}{G(\vec a;1+\sigma)G(\vec a;1-\sigma)}\Gamma(2\sigma)^{2K}\Gamma(1-2\sigma)^{2K}\ .
\ee
\end{theorem}
\end{shaded}
Define the discrete measure
\be
\nu\equiv \nu_{\vec a, K,\sigma;q} := q\sum_{k=0}^{\infty}w_{\vec a,K, \sigma}(k+\sigma)\delta_{k+\sigma}+q^{-1}\sum_{k=0}^{\infty}w_{\vec a,K, \sigma}(k-\sigma)\delta_{k-\sigma}\ ,
\ee
which depends on an extra complex parameter $q$ that will play an important book-keeping role in what follows. Here $\delta_\alpha$ stands for the Dirac measure concentrated at $\alpha\in \C$, and integrals with respect to $\nu$ are understood as
\be
\int f(z)d\nu(z) = q\sum_{k=0}^{\infty}w_{\vec a,K, \sigma}(k+\sigma)f(k+\sigma)+q^{-1}\sum_{k=0}^{\infty}w_{\vec a,K, \sigma}(k-\sigma)f(k-\sigma)\ .
\ee
By its definition, $\nu_{\vec a, K,\sigma;q}$ is supported on the union of two shifted half-lattices of positive integers:
\be
\operatorname{supp}(\nu_{\vec a, K,\sigma;q}) = (\N_{0}+\sigma)\cup (\N_{0}-\sigma)\ .
\ee
The moment generating function
\be
\psi(u) := \int e^{uz} \nu_{\vec a, K, \sigma;q}(z)
\ee
 can be written (see Lemma \ref{lemma:mgf} below) as a linear combination of generalized hypergeometric functions evaluated at $(-1)^pe^{u}$ as
\begin{align}
\nonumber
\psi(u)&=\frac{qe^{u \sigma}}{\Gamma(\vec a;K-\sigma)\Gamma(1+2\sigma)^2}
{}_p F_3\left(
\begin{array}{c} -\vec a-K+1+\sigma\\1+2\sigma\ \ 1+2\sigma\ \ 1
\end{array}; (-1)^{p}e^{u} \right)\\
&+\frac{q^{-1}e^{-u\sigma}}{\Gamma(\vec a;K+\sigma)\Gamma(1-2\sigma)^2}
{}_p F_3\left(
\begin{array}{c} -\vec a-K+1-\sigma\\1-2\sigma\ \ 1-2\sigma\ \ 1
\end{array}; (-1)^{p}e^{u} \right)\ ,
\end{align}
with the use of the condensed notation
\be
\label{eq:condensed}
-\vec a+s = [-a_1+s,\dots,-a_p+s]\ .
\ee
In particular, this means that  for $p \leq 3$ the moment generating function is an entire function of $u$ and hence all moments of $\nu$ are finite. If $p=4$ then
\be
|t| = |e^u|<1
\ee
needs to be assumed in general to have convergent expressions for the moments. For $p>4$ the moments do not exist unless there is a nonpositive integer appearing in the list $-\vec a+K+1+\sigma$, which, in effect, truncates the infinite sum.

The main result of the paper is the expression of the $K$-restricted partial sum ${\mathcal B}_{K}(\vec a;\sigma; t)$
in terms of the partition function of the discrete matrix model associated to $\nu_{K,\vec{a},\sigma;q}$:
\begin{shaded}
\begin{theorem}
\label{thm:matrix_integral}
The $K$-restricted partial sum ${\mathcal B}_{K}(\vec a;\sigma; t)$ is given by the expression
\begin{align}
\nonumber
&\sum_{\ell(\lambda)\leq K\atop \ell(\mu)\leq K}{\mathcal B}_{\lambda, \mu}(\vec a;\sigma)t^{|\lambda|+|\mu|}=\frac{Q(\vec a, K, \sigma)}{t^{K(K-1)}} \times\\
&\label{eq:restricted_sum_integral}
\times \frac{1}{(2K)!}\frac{1}{2\pi i} \oint_{|q|=1}\left[\underbrace{\int\!\cdots\! \int}_{2K} \prod_{1\leq i<j\leq 2K}(x_i-x_j)^{2}t^{\sum_{i=1}^{2K}{x_i}}\prod_{i=1}^{2K}d\nu(x_i)\right]\frac{dq}{q}
\end{align}
for the discrete measure $\nu=\nu_{K,\vec{a},\sigma;q}$, provided that $t$ is chosen such that the multiple integral in \eqref{eq:restricted_sum_integral} is convergent.
\end{theorem}
\end{shaded}
\begin{remark}
The contour integral in the variable $q$ in \eqref{eq:restricted_sum_integral} selects only those terms from the matrix model partition function that correspond to balanced  configurations, i.e.,  those with an equal number of points (eigenvalues) on  $\N_{0}+\sigma$ and $\N_{0}-\sigma$.
\end{remark}
\begin{remark} Note that the change of coordinates 
\be
t=e^{u}
\ee
allows $u$ to be interpreted as the first KP-Toda time as deformation parameter in the multiple integral \eqref{eq:restricted_sum_integral}.
\end{remark}
\begin{remark}
The discrete matrix model above is related to the model of Nekrasov and Okounkov for random partitions in {\rm\cite{nekrasov_okounkov}}, but the explicit partial summation formula \eqref{eq:restricted_sum_integral} does not appear there.
\end{remark}
As it is well-known, the multiple integral \eqref{eq:restricted_sum_integral} can be written in terms of the $2K\times 2K$ H\"ankel determinant built from the moments of the measure $\nu$, which can also be written as a Wronskian involving the moment generating function $\psi(u)$:
\begin{shaded}
\begin{corollary}
\label{cor:hankel}
The $K$-restricted partial sum ${\mathcal B}_K(\vec a;\sigma; e^{u})$ can be expressed as 
\begin{align}
&\sum_{\ell(\lambda)\leq K\atop \ell(\mu)\leq K}{\mathcal B}_{\lambda, \mu}(\vec a;\sigma)e^{u(|\lambda|+|\mu|)}\\
&\quad=\frac{Q(\vec a,K,\sigma)}{t^{K(K-1)}}
\frac{1}{2\pi i} \oint_{|q|=1}\left[\det\left(\psi^{(i+j-2)}(u)\right)_{i,j=1}^{2K}\right]\frac{dq}{q}\ ,
\end{align}
provided that $u$ is given such that all the derivatives $\psi^{k}(u)$ exist for $0\leq k\leq 2K-2$.
\end{corollary}
\end{shaded}
\begin{remark} A similar type of partial summation appears in {\rm \cite{bonelli_tanzini_zhao}}, where vortex partition functions lead to more general
AGT-type sums over pairs of partitions which are restricted to column partitions. It is shown in {\rm \cite{bonelli_tanzini_zhao}} that the corresponding columns-only partial sum can be written in terms of generalized hypergeometric functions. Their result is analogous to $K=1$ in our more general setting (up to conjugation of partitions).
\end{remark}

\noindent
{\bf Plan of the paper.} In Sec.~\ref{sec:interaction} it is shown how the combinatorial expression 
 defining ${\mathcal B}_{\lambda,\mu}(\vec a;\sigma)$ in \eqref{eq:B_lambda_mu} can be interpreted, after a suitable change of coordinates, as the exponential of the logarithmic energy of a coupled system of two sets of interacting particles in the complex plane (see Eq.~\eqref{eq:b}) whose configuration is labelled by the pair partitions $\lambda$ and $\mu$. The appearence of the Vandermonde factor in \eqref{eq:b} motivates the definition of the discrete matrix model introduced in Sec.~\ref{sec:matrix_model}. The simplest ``bare'' case of $P_{\mathrm{III}_3}$ is discussed briefly in Sec.~\ref{sec:PIII}, emphasizing the simple form the moment generating function $\psi(u)$ which does not depend on the maximal partition length $K$ in this case.
 
\section{The interaction of partitions in terms of particle coordinates}
\label{sec:interaction}
In this section the different ingredients of the coefficient ${\mathcal B}_{\lambda,\mu}(\vec a;\sigma)$ defined in  \eqref{eq:B_lambda_mu} will be rewritten step-by-step, using an alternative parametrization of the partitions $\lambda$ and $\mu$, and this leads to the proof of Theorem \ref{thm:b}.
\begin{definition} The \emph{particle locations} associated to $\lambda$ and $\mu$ are given by
\begin{align}
l_i &:=\lambda_i-i \qquad i \geq 1\\
m_i &:=\mu_i-i \qquad i \geq 1\ .
\end{align}
\end{definition}
It is easy to see that $l_1 > l_2 > \dots$ and $l_i = -i$ for sufficiently large $i$. Such sequences are in $1-1$ correspondence with partitions.
\begin{shaded}
\begin{proposition}
\label{prop:interact_half}
For a pair of partitions $\lambda$ and $\mu$ satisfying  $\ell(\lambda)\leq K$ and  $\ell(\mu)\leq K$, the following identity holds:
\be
\label{eq:interaction_term}
\prod_{(i,j)\in \lambda}(\lambda'_j+\mu_i-i-j+1+x)= \prod_{i=1}^{K}\frac{\Gamma(m_i+K+1+x)}{\Gamma(m_i-l_i+x)}\prod_{1\leq i\leq j\leq K}\frac{1}{(m_i-l_j+x)}\ .
\ee
\end{proposition}
\end{shaded}
\begin{proof} Let $1\leq i \leq \ell(\lambda)$ and consider the product of factors associated to the $i$th row of $\lambda$:
\be
\prod_{j=1}^{\lambda_i}(\lambda'_j+\mu_i-i-j+1+x) = \prod_{r=i}^{\ell(\lambda)}\prod_{j=\lambda_{r+1}+1}^{\lambda_r}(\lambda'_j+\mu_i-i-j+1+x)\ .
\ee
If $\lambda_{r+1}+1\leq  j\leq \lambda_r$ then $\lambda_j'=r$ and hence
\begin{align}
\prod_{j=\lambda_{r+1}+1}^{\lambda_r}(\lambda'_j+\mu_i-i-j+1+x) &= \prod_{j=\lambda_{r+1}+1}^{\lambda_r}(r+m_i-j+1+x)\\
&= \prod_{j=\lambda_{r+1}+1}^{\lambda_r}\frac{\Gamma(r+m_i-j+2+x)}{\Gamma(r+m_i-j+1+x)}\\
&= \frac{\Gamma(r+m_i-\lambda_{r+1}+1+x)}{\Gamma(r+m_i-\lambda_r+1+x)}\\
&= \frac{1}{m_i-l_r+x}\frac{\Gamma(m_i-l_{r+1}+x)}{\Gamma(m_i-l_r+x)}\ .
\end{align}
Therefore we have
\begin{align}
\prod_{j=1}^{\lambda_i}(\lambda'_j+\mu_i-i-j+1+x) &= \prod_{r=i}^{\ell(\lambda)}\frac{1}{m_i-l_r+x}\frac{\Gamma(m_i-l_{r+1}+x)}{\Gamma(m_i-l_r+x)}\\
&= \frac{\Gamma(m_i+\ell(\lambda)+1+x)}{\Gamma(m_i-l_i+x)}\prod_{r=i}^{\ell(\lambda)}\frac{1}{m_i-l_r+x}\ ,
\end{align}
where we used  $l_{\ell(\lambda)+1} = -\ell(\lambda)-1$. It is easy to see that
\be
\frac{\Gamma(m_i+\ell(\lambda)+1+x)}{\Gamma(m_i-l_i+x)}\prod_{r=i}^{\ell(\lambda)}\frac{1}{m_i-l_r+x}=\frac{\Gamma(m_i+K+1+x)}{\Gamma(m_i-l_i+x)}\prod_{i\leq j \leq K}\frac{1}{m_i-l_j+x}
\ee
for any $K \geq \ell(\lambda)$ because $l_j = -j$ for $j >\ell(\lambda)$. The full product over all boxes of the Young diagram of $\lambda$ is therefore equal to
\be
\prod_{(i,j)\in \lambda}(\lambda'_j+\mu_i-i-j+1+x)= \prod_{i=1}^{\ell(\lambda)}\left[
\frac{\Gamma(m_i+K+1+x)}{\Gamma(m_i-l_i+x)}\prod_{i\leq j \leq K}\frac{1}{m_i-l_j+x}
\right]\ .
\ee
To conclude \eqref{eq:interaction_term}, it is enough to note that if $i >\ell(\lambda)$ we have
\be
\frac{\Gamma(m_i+K+1+x)}{\Gamma(m_i-l_i+x)}\prod_{i\leq j \leq K}\frac{1}{m_i-l_j+x}=1\ .
\ee

\end{proof}
\begin{shaded}
\begin{proposition} Given a pair of partitions $\lambda,\mu$ and a positive integer $K$ such that 
\be
\ell(\lambda) \leq K\ , \quad \ell(\mu)\leq K\ ,
\ee
the following identity holds:
\begin{align}
\nonumber
&\frac{1}{\prod_{(i,j)\in \lambda}(\lambda'_j+\mu_i-i-j+1+x)^2\prod_{(i,j)\in \mu}(\mu'_j+\lambda_i-i-j+1-x)^2}\\
\label{eq:full_interaction}
&=\prod_{i=1}^{K}\frac{\Gamma(x)^2\Gamma(1-x)^2}{\Gamma(m_i+K+1+x)^2\Gamma(l_i+K+1-x)^2}\prod_{i, j=1}^{K}(m_i-l_j+x)^2\ .
\end{align}
\end{proposition}
\end{shaded}
\begin{proof} Proposition \ref{prop:interact_half} implies that
\begin{align}
\nonumber
&\frac{1}{\prod_{(i,j)\in \lambda}(\lambda'_j+\mu_i-i-j+1+x)^2\prod_{(i,j)\in \mu}(\mu'_j+\lambda_i-i-j+1-x)^2}\\
&=\prod_{i=1}^{K}\frac{\Gamma(m_i-l_i+x)^2\Gamma(l_i-m_i-x)^2(m_i-l_i+x)^2}{\Gamma(m_i+K+1+x)^2\Gamma(l_i+K+1-x)^2}\prod_{1\leq i ,j\leq K}(m_i-l_j+x)^2\ .
\end{align}
The formula \eqref{eq:full_interaction} now follows from the simple identity
\be
\Gamma(n+x)\Gamma(-n-x) = \frac{(-1)^{n+1}}{(n+x)}\Gamma(x)\Gamma(1-x)\ , \qquad n \in \N\cup\{0\}\ ,
\ee 
which can be easily seen from the recurrence relation satisfied by $\Gamma(z)$.
\end{proof}
\begin{definition}
The \emph{shifted particle locations} are
\begin{align}
L_i &:= l_i+K = \lambda_i-i+K &i\geq 1\ ,\\
M_i &:= m_i+K = \mu_i-i+K &i\geq 1\ ,
\end{align}
associated to the partitions $\lambda$ and $\mu$, respectively, where $K$ is a positive integer such that $\ell(\lambda) \leq K$ and $\ell(\mu)\leq K$.
\end{definition}
 
It is well known (see, e.g.~\cite{Fulton_Harris}) that the inverse of \emph{hook product} of $\lambda$ can be written as
\be
\frac{1}{\prod_{(i,j)\in \lambda}h_{\lambda}(i,j)} = \frac{\Delta\left(\{L_i\}_{i=1}^{K}\right)}{\prod_{i=1}^{K}L_i!}\ ,
\ee
where $\Delta$ stands for the Vandermonde product
\be
\Delta\left(\{L_i\}_{i=1}^{K}\right):=\prod_{1\leq i<j\leq K}(L_i-L_j)\ . 
\ee
This implies the following formula, once all necessary substitutions are made in \eqref{eq:full_interaction}.
\begin{lemma}
Given a pair of partitions $\lambda,\mu$ and a positive integer $K$ such that 
\be
\ell(\lambda) \leq K\ , \quad \ell(\mu)\leq K\ ,
\ee
the identity
\be
{\mathcal B}_{\lambda,\mu}([];\sigma) = \Gamma(2\sigma)^{2K}
\Gamma(1-2\sigma)^{2K}\prod_{1\leq i <j \leq 2K}(x_i-x_j)^{2}\prod_{i=1}^{2K}v(x_i)\ ,
\ee
is valid, where
\be
\label{eq:x}
\begin{split}
x_i &= L_i-\sigma\\
x_{K+i}&=M_i+\sigma
\end{split}
\qquad 1\leq i \leq K\ ,
\ee
and
\be
v(z)\equiv v_{\sigma}(z):=\frac{1}{\Gamma(z+1+\sigma)^2\Gamma(z+1-\sigma)^2}\ .
\ee
\end{lemma}

\begin{proof}
In terms of the parameters
\be
\begin{split}
x_i &= L_i-\sigma\\
x_{K+i}&=M_i+\sigma
\end{split}
\qquad 1\leq i \leq K\ ,
\ee
it is easy to see that
\be
\Delta\left(\{L_i\}_{i=1}^{K}\right)^2 \prod_{i,j=1}^{K}(M_i-L_j+2\sigma)^2\Delta\left(\{M_i\}_{i=1}^{K}\right)^2 = \prod_{1\leq i <j \leq 2K}(x_i-x_j)^{2}\ .
\ee
Moreover,
\begin{align}
\frac{1}
{\Gamma(L_i+1)^2\Gamma(L_i+1-2\sigma)^2}&=v_{\sigma}(L_i-\sigma)\\
\frac{1}
{\Gamma(M_i+1)^2\Gamma(M_i+1+2\sigma)^2}&=v_{\sigma}(M_i+\sigma)\ .
\end{align}
\end{proof}

\begin{proposition}
Let $\lambda$ be a partition such that $\ell(\lambda)\leq K$. Then
\be
\prod_{(i,j)\in \lambda}P(\vec a; i-j+\sigma)=\frac{G(\vec a;K+1+\sigma)}{G(\vec a;1+\sigma)}\prod_{i=1}^{K}\frac{1}{\Gamma(\vec a;-\l_i+\sigma)}\ .
\ee
\end{proposition}

\begin{proof}
Since 
\be
P(\vec a;z) = \frac{\Gamma(\vec a;z+1)}{\Gamma(\vec a;z)}\ ,
\ee
we have
\be
\prod_{j=1}^{\lambda_i}P(\vec a;i-j+\sigma) = \frac{\Gamma(\vec{a};i+\sigma)}{\Gamma(\vec{a};-l_i+\sigma)}\ .
\ee
Therefore, 
\be
\prod_{(i,j)\in \lambda}P(\vec a; i-j+\sigma)=\prod_{i=1}^{K}\frac{\Gamma(\vec a ;i+\sigma)}{\Gamma(\vec a;-\l_i+\sigma)}\ ,
\ee
where replacing the upper limit $\ell(\lambda)$ with any $K \geq \ell(\lambda)$ is justified, as seen above. To conclude the proof, note that
\be
\prod_{i=1}^{K}\Gamma(\vec a;i+\sigma) = \frac{G(\vec a;K+1+\sigma)}{G(\vec a;1+\sigma)}\ .
\ee

\end{proof}
\begin{proof}[Proof of Theorem \ref{thm:b}]
First, note that
\be
{\mathcal B}_{\lambda, \mu}(\vec a;\sigma)= {\mathcal B}_{\lambda, \mu}([];\sigma)\prod_{(i,j)\in \lambda}P(\vec a; i-j+\sigma)\prod_{(i,j)\in \mu}P(\vec a; i-j-\sigma)\ .
\ee
In terms of the parameters $(x_1,x_2,\dots, x_{2K})$ defined in \eqref{eq:x},
we have
\begin{align}
\prod_{(i,j)\in \lambda}P(\vec a; i-j+\sigma)&=\frac{G(\vec a;K+1+\sigma)}{G(\vec a;1+\sigma)}\prod_{i=1}^{K}\frac{1}{\Gamma(\vec{a};K-(L_i-\sigma))}\\
\prod_{(i,j)\in \mu}P(\vec a; i-j-\sigma)&=\frac{G(\vec a;K+1-\sigma)}{G(\vec a;1-\sigma)}\prod_{i=1}^{K}\frac{1}{\Gamma(\vec{a};K-(M_i+\sigma))}\ .
\end{align}
Since
\be
w_{\vec{a},K,\sigma}(z) = \frac{1}{\Gamma(\vec{a};K-z)}v_{\sigma}(z)\ ,
\ee
the proposed formula for ${\mathcal B}_{\lambda, \mu}(\vec a;\sigma)$ follows.
\end{proof}

\section{The discrete matrix model}
\label{sec:matrix_model}
In this section we prove that ${\mathcal B}_{K}(\vec a;\sigma;t)$ can be expressed in terms of the partition function of the matrix model associated to the discrete measure $\nu$.
\begin{proof}[Proof of Theorem \ref{thm:matrix_integral}]
In terms of the shifted particle locations $L_i$ and $M_i$ and the parameters
\be
\begin{split}
x_i&= L_i-\sigma\\
x_{K+i}&= M_i+\sigma
\end{split}
\qquad 1\leq i \leq K\ ,
\ee
the exponent of $t$ in the term associated to the pair $(\lambda, \mu)$ can be written as
\be
|\lambda|+|\mu| = \sum_{i=1}^{K}\left(L_i-\sigma +M_i+\sigma \right)-K(K-1)=\sum_{i=1}^{2K}x_i-K(K-1)\ .
\ee
The summation over pairs of partitions $(\lambda,\mu)$ such that $\ell(\lambda)\leq K$ and $\ell(\mu)\leq K$ is equivalent to the summing over all configurations
\be
L_1 > L_2 > \cdots > L_K \geq 0 \quad \mbox{and}\quad  M_1 > M_2 > \cdots > M_K \geq 0\ ,
\ee
and therefore, by Theorem \ref{thm:b}, the $K$-restricted partial sum  is equal to
\begin{align}
&\sum_{\ell(\lambda)\leq K\atop \ell(\mu)\leq K}{\mathcal B}_{\lambda, \mu}(\vec a;\sigma)t^{|\lambda|+|\mu|}\\
&=\frac{Q(\vec a, K, \sigma)}{t^{K(K-1)}}\sum_{L_1 > L_2 > \cdots > L_K \geq 0 \atop M_1 > M_2 > \cdots > M_K \geq 0}\prod_{1\leq i <j \leq 2K}(x_i-x_j)^{2}\prod_{i=1}^{2K}w(x_i)t^{\sum_{i=1}^{2K}x_i}\\
\nonumber
&=\frac{Q(\vec a, K, \sigma)}{t^{K(K-1)}}\frac{1}{K!^2}\sum_{L_1=0}^{\infty}\cdots \sum_{L_K=0}^{\infty}\sum_{M_1=0}^{\infty}\cdots \sum_{M_K=0}^{\infty}\Delta\left({\mathbf x}\right)^2 \prod_{i=1}^{2K}w(x_i)t^{\sum_{i=1}^{2K}x_i}\ .
\end{align}
The last equality is justified since the summand is a symmetric function of the indices $\{L_i\}_{i=1}^{K}$ and $\{M_i\}_{i=1}^{K}$ separately and it vanishes whenever $L_i=L_j$ or $M_i=M_j$ for some $i\not=j$ (guaranteed by the Vandermonde factor). Formally, the multiple sum can be written as a multiple integral in terms of the measure 
\be
\nu\equiv\nu_{\vec a, K,\sigma;q} = q\sum_{k=0}^{\infty}w_{\vec a,K, \sigma}(k+\sigma)\delta_{k+\sigma}+q^{-1}\sum_{k=0}^{\infty}w_{\vec a,K, \sigma}(k-\sigma)\delta_{k-\sigma}
\ee
as
\begin{align}
&\sum_{\ell(\lambda)\leq K\atop \ell(\mu)\leq K}{\mathcal B}_{\lambda, \mu}(\vec a;\sigma)t^{|\lambda|+|\mu|}\\
\label{eq:balance}
&=\frac{Q(\vec a, K, \sigma)}{t^{K(K-1)}}\frac{1}{K!^2}\underbrace{\int_{\N_{0}-\sigma}\cdots \int_{\N_{0}-\sigma}}_{K}\underbrace{\int_{\N_{0}+\sigma}\cdots \int_{\N_{0}+\sigma}}_{K}\Delta\left({\mathbf x}\right)^2 t^{\sum_{i=1}^{2K}x_i}\prod_{i=1}^{2K}d\nu(x_i)\ .
\end{align}
\begin{remark} Note that, despite the dependence of $\nu_{\vec a, K,\sigma;q}$ on $q$, the r.h.s.~of \eqref{eq:balance} does not depend on $q$:  both $q$ and $q^{-1}$ appears $K$ times when the integral is taken on $(\N_{0}-\sigma)^{K}\times( \N_{0}+\sigma)^{K}$.
\end{remark}
The unrestricted multiple integral
\be
\underbrace{\int_{\N_{0}\pm\sigma}\cdots \int_{\N_{0}\pm\sigma}}_{2K}\Delta\left({\mathbf x}\right)^2 t^{\sum_{i=1}^{2K}x_i}\prod_{i=1}^{2K}d\nu(x_i)
\ee
is a Laurent polynomial in the variable $q$, and the definition of the measure $\nu_{q,\sigma}$ implies that
\begin{align}
&\frac{1}{2\pi i} \oint_{|q|=1}\left[\underbrace{\int_{\N_{0}\pm\sigma}\cdots \int_{\N_{0}\pm\sigma}}_{2K}\Delta\left({\mathbf x}\right)^2 t^{\sum_{i=1}^{2K}x_i}\prod_{i=1}^{2K}d\nu(x_i)\right]\frac{dq}{q}\\
&=\binom{2K}{K}
\underbrace{\int_{\N_{0}-\sigma}\cdots \int_{\N_{0}-\sigma}}_{K}\underbrace{\int_{\N_{0}+\sigma}\cdots \int_{\N_{0}+\sigma}}_{K}\Delta\left({\mathbf x}\right)^2 t^{\sum_{i=1}^{2K}x_i}\prod_{i=1}^{2K}d\nu(x_i)\ ,
\end{align}
from which \eqref{eq:restricted_sum_integral} follows.
\end{proof}

\begin{lemma}
\label{lemma:mgf}
The moment generating function
\be
\psi(u) := \int e^{uz} \nu(z)
\ee
of the discrete measure $\nu= \nu_{\vec a, K, q,\sigma}$ can be written as a linear combination of generalized hypergeometric functions in $e^{u}$ as
\begin{align}
\nonumber
\psi(u)&=\frac{qe^{u \sigma}}{\Gamma(\vec a;K-\sigma)\Gamma(1+2\sigma)^2}
{}_p F_3\left(
\begin{array}{c} -\vec a-K+1+\sigma\\1+2\sigma\ \ 1+2\sigma\ \ 1
\end{array}; (-1)^{p}e^{u} \right)\\
&+\frac{q^{-1}e^{-u\sigma}}{\Gamma(\vec a;K+\sigma)\Gamma(1-2\sigma)^2}
{}_p F_3\left(
\begin{array}{c} -\vec a-K+1-\sigma\\1-2\sigma\ \ 1-2\sigma\ \ 1
\end{array}; (-1)^{p}e^{u} \right)\ ,
\end{align}
with the use of the condensed notation \eqref{eq:condensed}.
\end{lemma}

\begin{proof}
It is easy to see that the moment generating function of the discrete measure $\nu$ is
\begin{align}
\psi(u)&=\int e^{ux}d\nu_{K,\vec{a},\sigma,q}(x)\\ &= \sum_{k=0}^{\infty}\left[qw_{\vec a,K,\sigma}(k+\sigma)e^{u(k+\sigma)}+q^{-1}w_{\vec a,K,\sigma}(k-\sigma)e^{u(k-\sigma)}\right]\\
\label{eq:first_sum}
&=qe^{u\sigma}\sum_{k=0}^{\infty}\frac{e^{ku}}{\Gamma(\vec a;K-k-\sigma)\Gamma(k+1+2\sigma)^2\Gamma(k+1)^2}\\
&+q^{-1}e^{-u\sigma}\sum_{k=0}^{\infty}\frac{e^{ku}}{\Gamma(\vec a;K-k+\sigma)\Gamma(k+1-2\sigma)^2\Gamma(k+1)^2}\ .
\end{align}
Since
\begin{align}
&\frac{\Gamma(\vec a;K-k-\sigma)\Gamma(k+1+2\sigma)^2\Gamma(k+1)^2}{\Gamma(\vec a;K-k-1-\sigma)\Gamma(k+2+2\sigma)^2\Gamma(k+2)^2}\\
&\quad=\frac{P(\vec{a};K-k-1-\sigma)}{(k+1+2\sigma)^2(k+1)^2}\\
&\quad =\frac{(-1)^{p}\prod_{i=1}^{p}(k+1+\sigma-K-a_i)}{(k+1+2\sigma)^2(k+1)^2}\ ,
\end{align}
the first sum \eqref{eq:first_sum} can be written as
\begin{align}
&qe^{u\sigma}\sum_{k=0}^{\infty}\frac{e^{ku}}{\Gamma(\vec a;K-k-\sigma)\Gamma(k+1+2\sigma)^2\Gamma(k+1)^2}\\
&\ = \frac{qe^{u\sigma}}{\Gamma(\vec a;K-\sigma)\Gamma(1+2\sigma)^2}
{}_p F_3\left(
\begin{array}{c} -\vec a-K+1+\sigma\\ {[1+2\sigma,1+2\sigma, 1]}
\end{array}; (-1)^{p}e^{u} \right)\ ,
\end{align}
and, similarly,
\begin{align}
&q^{-1}e^{-u\sigma}\sum_{k=0}^{\infty}\frac{e^{ku}}{\Gamma(\vec a;K-k+\sigma)\Gamma(k+1-2\sigma)^2\Gamma(k+1)^2}\\
&\ = \frac{q^{-1}e^{-u\sigma}}{\Gamma(\vec a;K+\sigma)\Gamma(1-2\sigma)^2}
{}_p F_3\left(
\begin{array}{c} -\vec a-K+1-\sigma\\ {[1-2\sigma,1-2\sigma, 1]}
\end{array}; (-1)^{p}e^{u} \right)\ ,
\end{align}
which concludes the proof of the lemma.
\end{proof}

\section{Specialization to $P_{{III}'_3}$}
\label{sec:PIII}
As it was pointed out in Remark \ref{rem:decomp}, the isomonodromic $\tau$-function corresponding to $P_{\mathrm{{III}'_3}}$ is special since the corresponding discrete matrix model can be considered as an \emph{undressed} model, associated to the simplest choice $\vec a = [\ ]$. The measure $\nu$ corresponding to the empty list is particularly simple, leading to a less complicated formula for the $K$-restricted partial sum:
\begin{shaded}
\begin{corollary}
\begin{align}
\sum_{\ell(\lambda)\leq K\atop \ell(\mu)\leq K}{\mathcal B}_{\lambda, \mu}([\ ];\sigma)e^{u(|\lambda|+|\mu|)}&=\frac{\Gamma(2\sigma)^{2K}\Gamma(1-2\sigma)^{2K}}{t^{K(K-1)}}\times\\
&\times
\frac{1}{2\pi i} \oint_{|q|=1}\left[\det\left(\psi^{(i+j-2)}(u)\right)_{i,j=1}^{2K}\right]\frac{dq}{q}\ ,
\end{align}
where 
\begin{align}
\psi(u) &=\frac{qe^{u\sigma}}{\Gamma(1+2\sigma)^2}
{}_0 F_3\left(
\begin{array}{c} \\1+2\sigma\ \ 1+2\sigma\ \ 1
\end{array}; e^{u} \right)\\
&\quad +\frac{q^{-1}e^{-u\sigma}}{\Gamma(1-2\sigma)^2}
{}_0 F_3\left(
\begin{array}{c} \\1-2\sigma\ \ 1-2\sigma\ \ 1
\end{array}; e^{u} \right)\ .
\end{align}
\end{corollary}
\end{shaded}

\section{Conclusion}
We have found a closed expression of the $K$-restricted partial sum ${\mathcal B}_{K}(\vec{a};\sigma;t)$ of the conformal block ${\mathcal B}(\vec{a};\sigma;t)$ in terms of the partition function of an associated discrete matrix model.

This representation can be useful when the asymptotic behaviour of ${\mathcal B}_{K}(\vec{a};\sigma)$ is considered as $K \to \infty$, in light of the recent development on Riemann--Hilbert techniques for orthogonal polynomials with respect to discrete measures (see, e.g. \cite{bleher_liechty,discrete_op}), combined with techniques developed for the asymptotic analysis of random partitions (see \cite{eynard_all_order_partitions}).
It must be noted, however, that the measure $\nu$ is complex, and in order to perform a Riemann--Hilbert analysis for discrete orthogonal polynomials with complex weights one needs to overcome several difficulties. Also, to evaluate the limiting behaviour of the contour integral in \eqref{eq:restricted_sum_integral}, the asymptotics have to be uniform in $q$ as $K \to \infty$, which might lead to other complications.

Also, the representation of the length-restricted partial sums of the Painlev\'e conformal blocks are expressed in terms of Wronskians of generalized hypergeometric functions. This could allow to re-derive known particular solutions to Painlev\'e equations that are given in terms of their associated $\tau$-functions or to discover new solutions of similar nature finding the large $K$ limit of the conformal blocks and perfoming the summation in \eqref{eq:tau_expansion} (as seen in \cite{instanton_combinatorics}).

\subsection*{Acknowledgements} The author would like to thank D.~Guzzetti, D.~Yang, T.~Grava, B.~Dubrovin and A.~Tanzini for the helpful and stimulating discussions, and also the referees for their valuable comments and suggestions.

The present work was supported by the FP7 IRSES project RIMMP (\emph{Random and Integrable models in Mathematical Physics 2010-2014}),
the ERC project FroM-PDE (\emph{Frobenius Manifolds and Hamiltonian Partial Differential Equations  2009-13}) and the 
MIUR Research Project  \emph{Geometric and analytic theory of Hamiltonian systems in finite and infinite dimensions}.

\bibliography{painleve_tau_function_2}
\bibliographystyle{habbrv}
\end{document}